\documentclass{article}

\usepackage[english]{babel}

\usepackage[letterpaper,top=2cm,bottom=2cm,left=3cm,right=3cm,marginparwidth=1.75cm]{geometry}

\usepackage{amsmath}
\usepackage{graphicx}
\usepackage[colorlinks=true, allcolors=blue]{hyperref}
\usepackage{indentfirst} 
\usepackage{booktabs}
\usepackage{amsfonts}
\usepackage{amsmath}
\usepackage{multirow}
\usepackage{amsthm}
\usepackage{amssymb}
\usepackage{diagbox}
\usepackage{subcaption}
\usepackage{bm}
\usepackage{color}
\usepackage{graphicx}
\usepackage{natbib}
\usepackage{enumitem}
\usepackage{algorithm}
\usepackage{algorithmic}
\usepackage{mathtools}
\usepackage{tabularx}
\usepackage{xcolor}  % colors
\definecolor{pku-red}{RGB}{139,0,18}
\usepackage[capitalize,noabbrev]{cleveref}
\usepackage[capitalize,noabbrev]{cleveref}
\usepackage[textsize=tiny]{todonotes}
\theoremstyle{plain}
\newtheorem{theorem}{Theorem}[section]

\newtheorem{lemma}[theorem]{Lemma}

\theoremstyle{definition}
\newtheorem{definition}[theorem]{Definition}

\theoremstyle{remark}
\newtheorem{remark}[theorem]{Remark}

\newcommand{\EE}{\mathbb{E}}
\newcommand{\PP}{\mathbb{P}}
\newcommand{\RR}{\mathbb{R}}

\newcommand{\NN}{\mathbb{N}}
\newcommand{\LL}{\mathbb{L}}

\newcommand{\Ff}{\mathcal{F}}
\newcommand{\Vv}{\mathcal{V}}

\newcommand{\Ss}{\mathcal{S}}
\newcommand{\Nn}{\mathcal{N}}

\newcommand{\Dd}{\mathcal{D}}
\newcommand{\Mm}{\mathcal{M}}
\newcommand{\Uu}{\mathcal{U}}
\newcommand{\Xx}{\mathcal{X}}
\newcommand{\Zz}{\mathcal{Z}}
\newcommand{\Yy}{\mathcal{Y}}

\newcommand{\Pp}{\mathcal{P}}

\newcommand{\Ll}{\mathcal{L}}

\newcommand{\Rom}[1]{{\uppercase\expandafter{\romannumeral#1}}}

\newcommand{\name}{$\mathtt{CITransNet}$}
\newcommand{\mtt}[1]{$\mathtt{#1}$}

\usepackage{pgfplots} 
\usepgfplotslibrary{groupplots}

\title{A Context-Integrated Transformer-Based Neural Network for Auction Design}
\author{
\textbf{Zhijian Duan}$^{1}$, 
\textbf{Jingwu Tang}$^{1}$,
\textbf{Yutong Yin}$^{1}$, \\
\textbf{Zhe Feng}$^{2}$,
\textbf{Xiang Yan}$^{3}$,
\textbf{Manzil Zaheer}$^{4}$,
\textbf{Xiaotie Deng}$^{1}$ 
\\
$^{1}$ Peking University, Beijing, China \\
$^{2}$ Google Research, Mountain View, US 
$^{3}$ Google DeepMind, Mountain View, US \\
$^{4}$ Shanghai Jiao Tong University, Shanghai, China \\
\texttt{\{zjduan,tangjingwu,ytyin\}@pku.edu.cn},~ \texttt{zhef@google.com},\\
\texttt{xyansjtu@163.com},~ \texttt{manzilzaheer@google.com},~
\texttt{xiaotie@pku.edu.cn}
}
\date{}

\begin{document}
\maketitle

\begin{abstract}
One of the central problems in auction design is developing an incentive-compatible mechanism that maximizes the auctioneer's expected revenue.
While theoretical approaches have encountered bottlenecks in multi-item auctions, recently, there has been much progress on finding the optimal mechanism through deep learning. 
However, these works either focus on a fixed set of bidders and items, or restrict the auction to be symmetric. 
In this work, we overcome such limitations by factoring \emph{public} contextual information of bidders and items into the auction learning framework.
We propose $\mathtt{CITransNet}$, a context-integrated transformer-based neural network for optimal auction design, which maintains permutation-equivariance over bids and contexts while being able to find asymmetric solutions. 
We show by extensive experiments that $\mathtt{CITransNet}$ can recover the known optimal solutions in single-item settings, outperform strong baselines in multi-item auctions, and generalize well to cases other than those in training.
\end{abstract}

\section{Introduction}
Auction design is a classical problem in computational economics, with many applications on sponsored search~\citep{jansen2008sponsored}, resource allocation~\citep{huang2008auction} and blockchain~\citep{galal2018verifiable}.
Designing an incentive-compatible mechanism that maximizes the auctioneer's expected revenue is one of the central topics in auction design.
The seminal work by \citet{myerson1981optimal} provides an optimal auction design for the single-item setting; however, designing a revenue-optimal auction is still not fully understood even for two bidders and two items setting after four decades~\citep{dutting2019optimal}. 

Recently, pioneered by \citet{dutting2019optimal}, there is rapid progress on finding (approximate) optimal auction through deep learning, e.g.,~\citep{shen2018automated, luong2018optimal, tacchetti2019neural, nedelec2019adversarial, shen2020reinforcement,  brero2021reinforcement, liu2021neural}.
Typically, we can formulate auction design as a constrained optimization problem and find near-optimal solutions using standard machine learning pipelines. 
However, existing methods only consider simple settings: they either focus on a fixed set of bidders and items, e.g.~\citep{dutting2019optimal, rahme2020auction} or ignore the identity of bidders and items so that the auction is restricted to be symmetric~\cite{rahme2020permutation}.
As a comparison, in practice, auctions are much more complex beyond the aforementioned simple settings.
For instance, in e-commerce advertising, there are a large number of bidders and items (i.e., ad slots) with various features~\cite{liu2021neural}, and each auction involves a different number of bidders and items. 
To handle such a practical problem, we need a new architecture that can incorporate public features and take a different number of bidders and items as inputs.

\paragraph{Main Contributions}
In this paper, 
we consider contextual auction design, in which each bidder or item is equipped with context.
In contextual auctions, the bidder-contexts and item-contexts can characterize various bidders and items to some extent, making the auctions close to those in practice.
We formulate the contextual auction design as a learning problem and extend the learning framework proposed in \citet{dutting2019optimal} to our setting. 
Furthermore, we present a sample complexity result to bound the generalization error of the learned mechanism. 

To overcome the aforementioned limitations of the previous works,
we propose \name: a \textbf{C}ontext-\textbf{I}ntegrated \textbf{Trans}former-based neural \textbf{Net}work architecture as the parameterized mechanism to be optimized.
\name~incorporates the bidding profile along with the bidder-contexts and item-contexts to develop an auction mechanism.
It is built upon the transformer architecture~\cite{vaswani2017attention}, which can capture the complex mutual influence among different bidders and items in an auction.
As a result, \name~is permutation-equivariant~\citep{rahme2020permutation} over bids and contexts, i.e., any permutation of bidders (or items) in the bidding profile and bidder-contexts (or item-contexts) would cause the same permutation of auction result (We will provide a formal definition in \cref{remark:PE}).
Moreover, in \name, the number of parameters does not depend on the auction scale (i.e., the number of bidders and items), which brings \name~the potential of generalizing to auctions with various bidders or items, which we denote as \emph{out-of-setting generalization}.

We show by extensive experiments that \name~can almost reach the same result as \citet{myerson1981optimal} in single-item auctions and can obtain better performance in complex multi-item auctions compared to those strong baseline algorithms we use.
Additionally, we also justify its out-of-setting generalization ability.
Experimental results demonstrate that, under the same contextual setting, \name~can still perform well in auctions with a different number of bidders or items than those in training.

\paragraph{Further Related Work}
\label{sec:related_work}
As discussed before, it is an intricate task to design optimal auctions for multiple bidders and multiple items.
Many previous works focus on special cases (to name a few, \citet{manelli2006bundling,pavlov2011optimal,giannakopoulos2014duality,yao2017dominant,daskalakis2017strong,haghpanah2021pure}) and the algorithmic characterization of optimal auction (e.g., \citet{chawla2010multi,cai2012algorithmic,babaioff2014simple,yao2014n,cai2017simple,hart2017approximate}).
In addition, machine learning has also been applied to find approximate solutions for multiple items settings~ \citep{balcan2008reducing,lahaie2011kernel,dutting2015payment}, and there are also many works analyzing the sample complexity of designing optimal auctions~\citep{cole2014sample, devanur2016sample,balcan2016sample,guo2019settling,gonczarowski2021sample}.
In our paper, we follow the paradigm of \textit{automated mechanism design} \citep{conitzer2002complexity,conitzer2004self,sandholm2015automated}.

\citet{dutting2019optimal} propose the first neural network framework, \mtt{RegretNet}, to automatically design optimal auctions for general multiple bidders and multiple items settings by modeling an auction as a multi-layer neural network and using standard machine learning pipelines.
\citet{feng2018deep} and \citet{golowich2018deep} modify \mtt{RegretNet} to handle different constraints and objectives.
\citet{curry2020certifying} extend \mtt{RegretNet} to be able to verify strategyproofness of the auction mechanism learned by neural network.
\mtt{ALGNet}~\cite{rahme2020auction} models the auction design problem as a two-player game through parameterizing the misreporter as well.
\mtt{PreferenceNet}~\cite{peri2021preferencenet} encodes human preference (e.g. fairness) into \mtt{RegretNet}. \citet{rahme2020permutation} propose a permutation-equivariant architecture called \mtt{EquivariantNet} to design \emph{symmetric} auctions, a special case that is anonymous (bidder-symmetric) and item-symmetric.
In contrast, we study optimal contextual auction design, and our proposed \name~is permutation-equivariant while not restricted to symmetric auctions.

Existing literatures of contextual auction mainly discuss the online setting of some \emph{known} contextual repeated auctions, e.g., \emph{posted-price auctions}~\citep{amin2014repeated,mao2018contextual,drutsa2020optimal,zhiyanov2020bisection}, in which at every round the item is priced by the seller to sell to a strategic buyer, and \emph{second price auctions}~\citep{golrezaei2021dynamic}.
As a comparison, we consider the offline setting of contextual \emph{sealed-bid auction}.
We learn the mechanism from historical data and optimize the expected revenue for the auctioneer. 
Besides, we do not assume the conditional distribution of the bidder's valuation when given both the bidder-context and item-context.

\paragraph{Organization}
This paper is organized as follows: 
In \cref{sec:setting} we introduce contextual auction design, model the problem as a learning problem and derive a sample complexity for it;
In \cref{sec:net} we present the structure of \name, along with the training and optimization procedure;
We conduct experiments in \cref{sec:experiment} and draw the conclusion in \cref{sec:conclusion}.

\section{Contextual Auction Design}
\label{sec:setting}
In this section, we set up the problem of contextual auction design.
Then, we extend the learning framework proposed by~\citet{dutting2019optimal} to our contextual setting. 

\subsection{Contextual Auction}
We consider a contextual auction with $n$ bidders $N = \{1, 2, \dots, n\}$ and $m$ items $M = \{1, 2, \dots, m\}$.
Each bidder $i \in N$ is equipped with bidder-context ${x}_i \in \Xx \subset \RR^{d_{x}}$ and each item $j \in M$ is equipped with item-context ${y}_j \in \Yy \subset \RR^{d_y}$, in which $d_x$ and $d_y$ are the dimensions of bidder-context variables and item-context variables, respectively.
Denote ${{x}} = ({x}_1, {x}_2, \dots, {x}_n)$ as the bidder-contexts and ${{y}} = ({y}_1, {y}_2, \dots, {y}_m)$ as the item-contexts.
${{x}}$ and ${{y}}$ are sampled from underlying joint probability distribution $\Dd_{{x}, {{y}}}$.
Let $v_{ij}$ be the valuation of bidder $i$ for item $j$. Conditioned on bidder-context ${x}_i$ and item-context ${y}_j$, $v_{ij}$ is sampled from a distribution $\Dd_{v_{ij}|{x}_i, {y}_j}$, i.e., the distribution of $v_{ij}$ depends on both $x_i$ and $y_j$.

The valuation profile $v = (v_{ij})_{i\in N, j\in M} \in \RR^{n\times m}$ is unknown to the auctioneer, however, she knows the sampled bidder-contexts ${{x}}$ and item-contexts ${{y}}$.
In this paper, we only focus on \emph{additive} valuation setting, i.e., the valuation of each bidder $i$ for a set of items $S\subseteq M$ is the sum of valuation for each item $j\in S$: $v_{iS} = \sum_{j \in S}v_{ij}$.
At an auction round, each bidder bids for each item.
Given the bidding profile (or bids) $b = (b_{ij})_{i \in N, j\in M}$, the contextual auction mechanism is defined as follows:

\begin{definition}[Contextual Auction Mechanism]
A contextual auction mechanism $(g, p)$ consists of an allocation rule $g$ and a payment rule $p$:
\begin{itemize}
% [itemsep=2pt,topsep=0pt,parsep=0pt]
    \item The allocation rule $g=(g_{ij})_{i\in N, j\in M}$, in which $g_{ij}\colon \RR^{n\times m} \times \Xx^n\times {\Yy}^m \rightarrow [0,1]$ computes the probability that item $j$ is allocated to bidder $i$, given the bidding profile $b\in \RR^{n\times m}$, bidder-contexts $ {x}\in \Xx^n$ and item-contexts $ {y} \in {\Yy}^m$. 
    For all $b, {{x}}, {{y}}$, and $j \in M$, we have $\sum_{i=1}^n g_{ij}(b, {{x}}, {{y}})\leq 1$ to guarantee no item is allocated more than once. 
    
    \item The payment rule $p=(p_1,p_2,\dots,p_n)$, in which $p_i\colon \mathbb{R}^{n\times m}\times \Xx^n\times {\Yy}^m \rightarrow \mathbb{R}_{\geq 0}$ computes the price bidder $i$ need to pay, given the bidding profile $b\in \RR^{n\times m}$, bidder-contexts $ {x}\in \Xx^n$ and item-contexts $ {y} \in {\Yy}^m$. 
\end{itemize}
\end{definition}

Define $\Vv = \Vv_1\times \Vv_2\times \dots \times \Vv_n$ be the joint valuation profile domain set, in which $\Vv_i$ is the domain set of all the possible valuation profiles $v_i = (v_{i1}, v_{i2}, \dots, v_{im})$ of bidder $i$.
Let $\Vv_{-i} = (\Vv_1, \dots, \Vv_{i-1}, \Vv_{i+1}, \dots, \Vv_n)$ be the joint valuation profile domain set except $\Vv_i$.
Similarly, we denote $v_{-i} = (v_1, \dots, v_{i-1}, v_{i+1}, \dots, v_n)$ and $b_{-i} = (b_1, \dots, b_{i-1}, b_{i+1}, \dots, b_n)$.
Without loss of generality, we assume $b_i \in \Vv_i$ for all $i\in N$. 
Each bidder $i\in N$ aims to maximize her utility, defined as follows,
\begin{definition}[Quasilinear utility]
In an additive valuation auction setting, the utility of bidder $i$ under mechanism $(g, p)$ is defined by
\vspace{-10pt}
\begin{equation*}
    u_i({v}_i, b,{{x}},{{y}})= \sum_{j=1}^m g_{ij}(b,{{x}},{{y}}) v_{ij} - p_i(b,{{x}},{{y}})
\end{equation*}
for all $ v_i \in \Vv_i, b\in \Vv, {x}\in \Xx^n, {y}\in {\Yy}^m$.
\end{definition}

In this work, we want the auction mechanism to be \emph{dominant strategy incentive compatible} (DSIC)\footnote{There is another weaker notion of incentive compatibility, Bayesian incentive compatibility (BIC), in the literature. In practice, DSIC is more desirable than BIC. It doesn't require prior knowledge of the other bidders and is more robust. In this work, we only focus on DSIC, similar to~\citet{dutting2019optimal}.}, defined as below,  
\begin{definition}[DSIC]
\label{def:DSIC}
An auction $(g, p)$ is \emph{dominant strategy incentive compatible} (DSIC) if for each bidder, the optimal strategy is to report her true valuation no matter how others report.
Formally, for each bidder $i\in N$, for all ${{x}} \in \Xx^n,{{y}} \in {\Yy}^m$ and for arbitrary $b_{-i}\in \Vv_{-i}$, we have
\begin{equation*}
\begin{aligned}
u_i({v}_i,({v}_i,b_{-i}),{{x}},{{y}})) \ge u_i({v}_i,({b}_i,b_{-i}),{{x}},{{y}})),
\end{aligned}
\end{equation*}
for all $b_i \in \Vv_i$.
\end{definition}

Besides, the auction mechanism needs to be \emph{individually rational} (IR), defined as follows,
\begin{definition}[IR]
\label{def:IR}
An auction $(g, p)$ is \emph{individually rational} (IR) if for each bidder, truthful bidding will receive a non-negative utility.
Formally, for each bidder $i \in N$, for all ${{x}} \in \Xx^n, {{y}} \in {\Yy}^m$ and for arbitrary ${v}_i \in \Vv_i, b_{-i} \in \Vv_{-i}$, we have
\begin{equation}\label{eq:IR_eq}\tag{IR}
  u_i({v}_i,({v}_i,b_{-i}),{{x}},{{y}})\geq 0.  
\end{equation}
\end{definition}

In a DSIC and IR auction, rational bidders would truthfully report their valuations.
Therefore, let $\Dd_{v, {{x}}, {{y}}}$ be the joint distribution of $v$, ${{x}}$ and ${{y}}$, the expected revenue is:
\begin{equation*}
\begin{aligned}
    rev := & \EE_{(v,{{x}},{{y}}) \sim \Dd_{v, {{x}},{{y}}}}
    \left[\sum_{i=1}^n p_i(v,{{x}},{{y}})\right]. 
\end{aligned}
\end{equation*}
Optimal contextual auction design aims to find an auction mechanism that maximizes the expected revenue while satisfying the DSIC and IR conditions.
  
\subsection{Contextual Auction Design as a Learning Problem}

Similar to~\citet{dutting2019optimal}, we formulize the problem of optimal auction design as a learning problem.
First, we define \textit{ex-post regret}:
\begin{definition}[(Ex-post) Regret]
The ex-post regret for a bidder $i$ under mechanism $(g, p)$ is the maximum utility gain she can achieve by misreporting when the bids of others are fixed, i.e., 
\begin{equation*}
% \tag{Regret}
\begin{aligned}
    rgt_{i}(v,{{x}},{{y}}) := \max_{{b}_i\in \Vv_i}& u_i({v}_i, ({b}_i,v_{-i}),{{x}},{{y}})
    -u_i({v}_i, v,{{x}},{{y}}).
\end{aligned}
\end{equation*}
\end{definition}
In particular, similar to~\citet{dutting2019optimal}, the DSIC condition is equivalent to
% \begin{equation}
% \label{eq:DSIC=regret}
    $rgt_{i}(v,{{x}},{{y}}) = 0, \forall i\in N, v\in \Vv, x\in \Xx^n,y \in \Yy^m$.
% \end{equation}
By assuming that $\Dd_{v,x,y}$ has full support on the space of $(v, x, y)$ and recognizing that the regret is non-negative, an auction satisfies DSIC (except for measure zero events) if 
\begin{equation}\label{eq:regretDSIC}
\tag{DSIC} 
\EE_{(v,x,y)\sim \Dd_{v,x,y}}\left[\sum_{i=1}^nrgt_{i}(v,{{x}},{{y}})\right]=0.
\end{equation}
Let $\Mm$ be the set of all the auction mechanisms that satisfy \cref{eq:IR_eq}.
By setting~\cref{eq:regretDSIC} as a constraint, we can formalize the problem of finding an optimal contextual auction as a constraint optimization:
\begin{equation}
\tag{\Rom{1}}\label{eq:exact_prob}
\begin{aligned}
\min_{(g,p)\in\Mm}& 
-\EE_{(v,{{x}},{{y}}) \sim \Dd_{v, {{x}},{{y}}}}
\left[\sum_{i=1}^n p_i(v,{{x}},{{y}})\right]
\\
\mathrm{s.t.~} & \EE_{(v,x,y)\sim \Dd_{v,x,y}}\left[\sum_{i=1}^nrgt_{i}(v,{{x}},{{y}})\right]=0.
\end{aligned}
\end{equation}

This optimization problem is generally intractable due to the intricate constraints\footnote{In the automated mechanism design literature~\cite{conitzer2002complexity,conitzer2004self}, \cref{eq:exact_prob} can be formulated as a linear programming. However, this LP is hard to solve in practice because of the exponential number of constraints, even for discrete value distribution settings.}.
To handle such a problem, we parameterize the auction mechanism as $(g^w, p^w)$, where $w \in \RR^{d_w}$ are the parameters (with dimension $d_w$) to be optimized. 
All the expectation terms are computed empirically by $L$ samples of $(v, x, y)$ independently drawn from $\Dd_{v,{{x}},{{y}}}$.
The empirical ex-post regret for bidder $i$ under parameters $w$ is defined as
\begin{equation} 
\label{eq:emp_reg}
% \tag{$\widehat{\mbox{Regret}}$}
\begin{aligned}
\widehat{rgt}_i(w) :=& \frac{1}{L}\sum_{\ell=1}^L rgt_i^w(v^{(\ell)}, x^{(\ell)}, y^{(\ell)}), 
\end{aligned}
\end{equation}
where $rgt_i^w(v, x, y)$ is computed based on the parameterized mechanism $(g^w,p^w)$.
On top of that, the learning formulation of \cref{eq:exact_prob} is
\begin{equation}
\label{eq:empirical_prob}
\tag{\Rom{2}}
\begin{aligned}
\min_{w\in \RR^{d_w}} &-\frac{1}{L}\sum_{\ell=1}^L \sum_{i=1}^n p_i^w(v^{(\ell)},{{x}}^{(\ell)},{{y}}^{(\ell)})
\\
&\text{s.t.}\quad\widehat{rgt}_i(w)=0, \forall i \in N
\end{aligned}
\end{equation}
 
\cref{eq:IR_eq} can be satisfied through the architecture design.
See \cref{sec:net:output} for the discussion.

\subsection{Sample Complexity}
\label{subsec:theory}
We provide a sample complexity to bound the two gaps at the same time: the gap between empirical revenue and expected revenue, and the gap between empirical regret and expected regret.
Such result justifies the feasibility to approximately solve \cref{eq:exact_prob} by \cref{eq:empirical_prob}.

For contextual auction mechanism class $\Mm$, similar to \citet{dutting2019optimal}, we measure the capacity of $\Mm$ via \textit{covering numbers}~\citep{shalev2014understanding}. We define the $\ell_{\infty, 1}$-distance between two auction mechanisms $(g, p), (g', p')\in \Mm$ as $\max_{v,x,y} \sum_{i\in N,j\in M}|g_{ij}(v, x, y) - g'_{ij}(v, x, y)| + \sum_{i\in N}|p_i(v, x, y) - p'_i(v, x, y)|$.
For all $r > 0$, let $\Nn_{\infty, 1}(\Mm, r)$ be the minimum number of balls with radius $r$ that cover all the mechanisms in $\Mm$ under $\ell_{\infty, 1}$-distance (called the $r$-covering number of $\Mm$). 
We have the following result:

\begin{theorem}
\label{thm:UC}
For each bidder $i$, assume w.l.o.g.~that the valuation function $v_i$ satisfies $v_i(S) \leq 1,\, \forall S \subseteq M$.
Fix $\delta, \epsilon \in (0,1)$, for any $(g^w,  p^w) \in \Mm$, 
when 
\begin{equation*}
L \ge \frac{9n^2}{2\epsilon^2}\left(\ln\frac{4}{\delta} + \ln{ \Nn_{\infty, 1}(\Mm,\frac{\epsilon}{6n})}\right),
\end{equation*}
with probability at least $1 - \delta$ 
over draw of training set $S$ of $L$ samples from $\Dd_{v,{x},{y}}$, we have both
\begin{equation}
\label{eq:UC:profit}
\begin{aligned}
\left|\sum_{i=1}^n\Big(\EE_{(v,x,y)}p^w_i(v, {x}, {y}) - \sum_{\ell=1}^L \frac{p^w_i(v^{(\ell)}, {x}^{(\ell)}, {y}^{(\ell)})}{L}\Big)\right| \le \epsilon,
\end{aligned}
\end{equation}
and 
\begin{equation}
\label{eq:UC:regret}
\begin{aligned}
\bigg|\EE_{(v,{x},{y}) \sim \Dd_{v,{x},{y}}}\Big[\sum_{i=1}^n rgt_i^w(v,x,y)\Big]  - 
\sum_{i=1}^n \widehat{rgt}_i(w)\bigg| \le \epsilon.
\end{aligned}
\end{equation}
\end{theorem}

See \cref{app:proof:thm:UC} for detailed proofs.

\section{Model Architecture}
\label{sec:net}

\begin{figure*}[t]
    \centering
    \includegraphics[width=\textwidth,trim=0 36mm 0 6mm,clip]{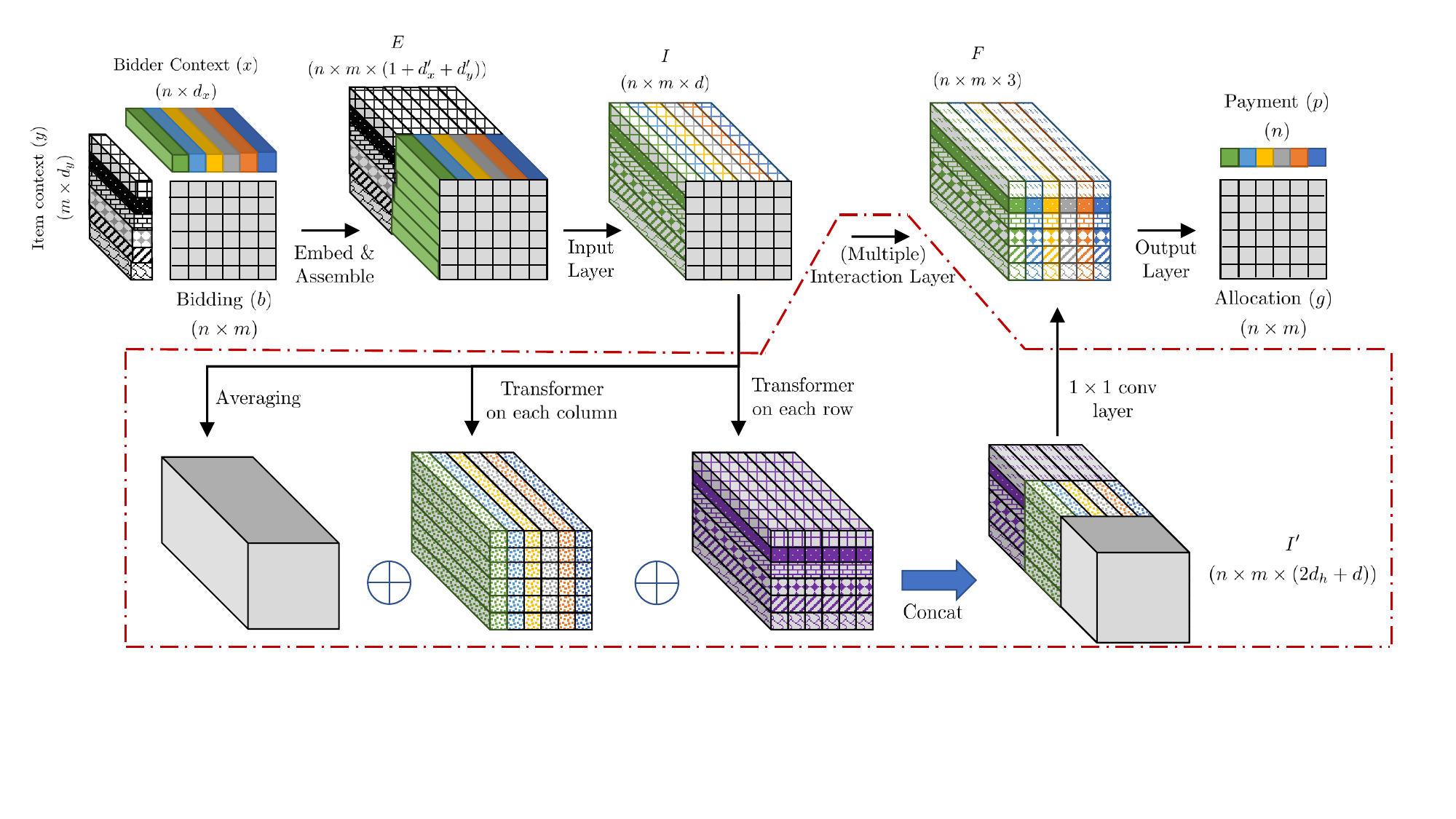}
    \caption{
    A schematic view of \name, which
    takes the bidding profile $b \in \RR^{n\times m}$, bidder-contexts ${x} \in \Xx^n$ and item-contexts ${y} \in \Yy^m$ as inputs.
    We first embeds $x$ and $y$ into $e_x \in \RR^{d'_x}$ and $f_y \in \RR^{d'_y}$, and then assemble $e_x$, $f_y$ and $b$ into $E\in \RR^{n\times m\times (1+d'_x+d'_y)}$, the initial representation for each bidder-item pair.
    The remaining part of our input layer along with one or more transformer-based interaction layers are adopted to model the mutual interactions among different bidders and items.
    Based on the output $F \in \RR^{n\times m\times 3}$ of the last interaction layer, we compute the allocation and payment result via the final output layer.
    }
    \label{fig:model}
    % \vspace{-2mm}
\end{figure*}

In this section, we describe \name, the proposed context-integrated transformer-based neural network for computing allocation and payment in \cref{eq:empirical_prob}.

\subsection{Overview of \name}
As shown in \cref{fig:model}, \name~takes the bidding profile $b \in \RR^{n\times m}$, bidder-contexts ${x}$ and item-contexts ${y}$ as inputs.
An input layer is used first to compute a $d$-dimensional feature vector for each bidder-item pair.
Afterward, the features of all the bidder-item pairs, i.e., $I \in \RR^{n\times m\times d}$, are fed into one or multiple interaction layers.
Such transformer-based interaction layers model the interactions between bidders and items.
The global feature maps $F \in \RR^{n\times m\times 3}$ are obtained through the last interaction layer.
Finally, we compute the allocation result $g^w(b, {x}, {y})$ and payment result $p^w(b, {x}, {y})$ through the final output layer.

\subsection{Input Layer}

First, we apply a pre-processing to obtain a representation $e_{x_i} \in \RR^{d'_x}$ for each bidder context $x_i$ and $f_{y_j} \in \RR^{d'_y}$ for each item context $y_j$:
\begin{itemize}
% [itemsep=2pt,topsep=0pt,parsep=0pt]
    \item If $x_i$ (or $y_j$) is drawn from a continuous space, simply set $e_{x_i} = x_i$ (or $f_{y_j} = y_j$).
    \item If $x_i$ (or $y_j$) is only drawn from some finite types, embed it into a continuous space, similarly as the common procedure in word embedding~\cite{mikolov2013distributed}. The corresponding embedding is $e_{x_i}$ (or $f_{y_j}$).
\end{itemize}

We construct the initial representation for each bidder-item pair: $E = (E_{i,j})_{i\in N, j\in M}$, in which
\begin{equation*}
    E_{ij} = [b_{ij};{e}_{x_i};f_{y_j}] \in \RR^{1 + d'_x + d'_y},
\end{equation*}

Afterwards, two $1\times 1$ convolutions with a ReLU activation are applied to $E$ and reduce the third-dimension of $E$ from $1 + d_x' + d_y'$ to $d - 1$. Formally, 
\begin{equation*}
    E' = \mathrm{Conv}_2(\mathrm{ReLU}(\mathrm{Conv}_1(E))) \in \RR^{n\times m\times (d-1)},
\end{equation*}
where both $\mathrm{Conv}_1$ and $\mathrm{Conv}_2$ are $1\times1$ convolutions, and $\mathrm{ReLU}(x):=\max(x, 0)$.
By concatenating $E'$ and the bids $b$, we get $I \in \RR^{n\times m\times d}$, the output of our input layer:
\begin{equation*}
\begin{aligned}
    I &= [b; E'] \in \RR^{n\times m \times d},
\end{aligned}
\end{equation*}
where feature $ {I}_{ij} \in \RR^d$ in $I$ captures the bidding and context information of the corresponding bidder-item pair.

\subsection{Interaction Layer}
Given the representation for all bidder-item pairs $I \in \RR^{n\times m\times d}$, we move on to model the interactions between different bidders and items, which is illustrated in the lower part of \cref{fig:model}.
The interaction layer is built based upon transformer model~\cite{vaswani2017attention}, which can be used to capture the high-order feature interactions of input through the multi-head self-attention module~\cite{song2019autoint}.
See \cref{app:transformer} for a description of transformer.

Precisely, for each bidder $i$, we model its interactions with all the $m$ items through transformer on the $i$-th row of $I$ (denoted as ${I}_{i, \cdot} \in \RR^{m\times d_h}$):
\begin{equation*}
    {I}^{\mathrm{row}}_{i, \cdot} = \mathrm{transformer}({I}_{i, \cdot}) \in \RR^{m\times d_h}, \forall i \in N,
\end{equation*}
where $d_h$ is the size of the hidden nodes in the MLP part of the transformer.
Symmetrically, for each item $j$, we model its interactions with all the $n$ bidders through another transformer on the $j$-th column of $I$ (called ${I}_{\cdot, j} \in \RR^{n\times d_h}$):
\begin{equation*}
    {I}^{\mathrm{column}}_{\cdot, j} = \mathrm{transformer}({I}_{\cdot, j}) \in \RR^{n\times d_h}, \forall j \in M.
\end{equation*}
Afterward, the global representation for all the bidder-item pairs is obtained by the average of all the features
\begin{equation*}
     e^{\mathrm{global}} = \frac{1}{nm}\sum_{i=1}^n\sum_{j=1}^m {I}_{ij} \in \RR^d.
\end{equation*}

Combining $I^{\mathrm{row}}, I^{\mathrm{column}}$ and $ e^{\mathrm{global}}$ together, we get new features $I'_{ij}$ for each bidder-item pair
\begin{equation*}
    {I}'_{ij} := [{I}^{\mathrm{row}}_{ij}; {I}^{\mathrm{column}}_{ij};  e^{\mathrm{global}}] \in \RR^{2d_h + d}
\end{equation*}
Finally, as what we did in input layer, 
two $1\times 1$ convolutions with a ReLU activation are applied on $I'$ in order to reduce the third dimension of $I'$ 
from $2d_h + d$ to $d_{\mathrm{out}}$.
Formally,
\begin{equation*}
    F = \mathrm{Conv}_4(\mathrm{ReLU}(\mathrm{Conv}_3(I'))) \in \RR^{n\times m\times d_{\mathrm{out}}},
\end{equation*}
where both $\mathrm{Conv}_4$ and $\mathrm{Conv}_3$ are $1\times 1$ convolutions, and $F$ is the output of the interaction layer.
By stacking multiple interaction layers, we can model higher-order interactions among all the bidders and items.

\subsection{Output Layer} 
\label{sec:net:output}
In the last interaction layer, we set $d_{\mathrm{out}}=3$ and get the global feature maps $F = (F^h, F^q, F^p) \in \RR^{n\times m \times 3}$, which will be used to compute the final allocation and payment in the output layer.

The first feature map $F^h \in \RR^{n\times m}$ is used to compute the original allocation probability $h^w(b,{x},{y}) \in [0, 1]^{n\times m}$ by softmax activation function on each column of $F^h$, i.e., 
\begin{equation*}
    h^w_{\cdot,j} = \mathrm{Softmax}(F^h_{\cdot, j}), \forall j \in M.
\end{equation*}
Here $h^w_{i,j}$ is the probability that item $j$ is allocated to bidder $i$ and we have $\sum_{i=1}^nh^w_{i,j} = 1$ for each item $j \in M$.

Since some item $j$ may not be allocated to any bidder, we use the second feature map $F^q$ to adjust $h_w$.
The weight $q^w(b,{x},{y}) \in (0,1)^{n\times m}$ of each probability is computed through sigmoid activation on $F^q$:
\begin{equation*}
    q^w_{i,j} = \mathrm{Sigmoid}(F^q_{i,j}), \forall i\in N, \forall j \in M,
\end{equation*}
where $\mathrm{Sigmoid}(x):= \frac{1}{1+e^{-x}} \in (0, 1)$.

The allocation result $g^w$ is then obtained by combining $h^w$ and $q^w$ together:
\begin{equation*}
g^w_{ij}(b,{x},{y}) = q^w_{ij}(b,{x},{y})h^w_{ij}(b,{x},{y}).    
\end{equation*}
As a result, we have $0 < \sum_{i=1}^ng^w_{i,j}(b,{x},{y}) < 1$ for each item $j \in M$.

For payment, we compute payment fraction $\tilde p^w(b,{x},{y}) \in (0, 1)^n$ via the third feature map $F^p$:
\begin{equation*}
    \tilde p^w_i = \mathrm{Sigmoid}\big(\frac{1}{m}\sum_{j=1}^m F^p_{ij}\big), \forall i\in N,
\end{equation*}
where $\tilde p^w_i$ is the fraction of bidder $i$'s utility that she has to pay to the auctioneer. 
Given the allocation $g^w$ and payment fraction $\tilde{p}^w$, the payment for bidder $i$ is 
\begin{equation*}
    p^w_i(b,{x},{y}) = \tilde p^w_i(b,{x},{y})\sum_{j=1}^mg^w_{ij}(b,{x},{y})b_{ij}.
\end{equation*}
By doing so, Equation \eqref{eq:IR_eq} is satisfied.

\begin{remark}[Permutation-equivariant]
\label{remark:PE}
Similar to the definition in \citet{rahme2020permutation}, we say an auction mechanism $(g^w, p^w)$ is permutation-equivariant if for any two permutation matrices $\Pi_{n}\in \{0,1\}^{n\times n}$ and $\Pi_{m}\in \{0,1\}^{m\times m}$, and any input (including bids $b \in \mathbb{R}^{n\times m}$, bidder-contexts $x \in \mathbb{R}^{n\times d_x}$ and item-contexts $y \in \mathbb{R}^{m\times d_y}$), we have $g^w(\Pi_{n}b\Pi_{m}, \Pi_n x, \Pi_m^T y)=\Pi_{n}g^w(b,x,y)\Pi_{m}$ and $p^w(\Pi_{n}b\Pi_{m}, \Pi_n x, \Pi_m^T y)=\Pi_{n}p^w(b,x,y)$.
Transformer is known to be permutation-equivariant, since it maps each embedding in input to a new embedding that incorporates the information of the \emph{set} of all the input embeddings.
Moreover, the $1\times 1$ convolutions we use in \name~are all per bidder-item wise, i.e., acting on each bidder-item pair.
As a result, \name~maintains permutation-equivariant.
\end{remark}
\subsection{Optimization and training}
Similar to~\citet{dutting2019optimal}, \name~is optimized through the augmented Lagrangian method.
The Lagrangian with a quadratic penalty is:
\begin{equation}
\label{eq:loss}
\begin{aligned}
    \mathcal{L}_{\rho}(w;\lambda) = &-\frac{1}{L}\sum_{\ell=1}^L \sum_{i=1}^n p_i^w(v^{(\ell)}, {x}^{(\ell)}, {y}^{(\ell)})~ +
    \\
    &\sum_{i=1}^n\lambda_i\widehat{rgt}_i(w)+\frac{\rho}{2}\sum_{i=1}^n\left(\widehat{rgt}_i(w)\right)^2, 
\end{aligned}
\end{equation}

where $\lambda = (\lambda_1, \lambda_2, \dots, \lambda_n) \in\mathbb{R}^n$ is the Lagrange multipliers, and $\rho>0$ is a hyperparameter that controls the weight of the quadratic penalty. 
During optimization, we update the model parameters and Lagrange multipliers in turn, i.e., we alternately find $w^{new}\in\arg\min_w \mathcal{L}_{\rho}(w^{old},\lambda^{old})$ and update $\lambda_i^{new}=\lambda_i^{old}+\rho\cdot \widehat{rgt}_i(w^{new}), \forall i \in N$. 
See \cref{app:algorithm} for a detailed optimization and training procedure.

\section{Experiments}
\label{sec:experiment}
In this section, we conduct empirical experiments to show the effectiveness of \name~in different contextual auctions
\footnote{
Our implementation is available at \url{https://github.com/zjduan/CITransNet}.
}.
Afterward, we demonstrate the out-of-setting generalization ability for \name~by evaluating the trained model in settings with different numbers of bidders or items. 
Our experiments are run on a Linux machine with NVIDIA Graphics Processing Unit (GPU) cores.
Each result is obtained by averaging across $5$ different runs. 
We ignore the standard deviation since it is small in all the experiments.

\paragraph{Baseline Methods} We compare \name~with the following baselines: 
1)
\mtt{Item\text{-}wise\;Myerson}, a strong baseline used in~\citet{dutting2019optimal}, which independently applies Myerson auction with respect to each item
\footnote{
\mtt{Bundle~Myerson} is another baseline used in~\citet{dutting2019optimal} that satisfies both DSIC and IR. 
However, we find it always performs worse than \mtt{Item\text{-}wise\;Myerson}, both in our experiments and in~\citet{dutting2019optimal}.
Therefore, we do not present its results.
};
2) \mtt{RegretNet}~\citep{dutting2019optimal}, which adopts fully-connected neural networks to compute auction mechanism;
\mtt{EquivariantNet}~\citep{rahme2020permutation}, which is a permutation-equivariant architecture to design the special mechanism of symmetric auctions
\footnote{
While \citet{rahme2020auction} formulate auction learning as an adversarial learning framework, we view this as an orthogonal problem since this work mainly focuses on the innovation of neural architectures. Therefore, to make a fair comparison, we adopt the learning framework in \citet{dutting2019optimal} for the baselines and leave the adversarial learning framework extension for future work.
};
3) \mtt{CIRegretNet} and \mtt{CIEquivariantNet}, which are the context-integrated version of  \mtt{RegretNet} and \mtt{EquivariantNet}.
Specifically, we replace the interaction layers of our \name~with \mtt{RegretNet} and \mtt{EquivariantNet}, respectively.
We set these baselines to evaluate the effectiveness of our transformer-based interaction layers.

See \cref{app:implement} for implementation details of all methods.

\paragraph{Evaluation} 
Following \citet{dutting2019optimal} and \citet{rahme2020permutation}, to evaluate each method, we adopt empirical revenue (the minus objective in \cref{eq:empirical_prob}) and empirical ex-post regret average across all the bidders $\widehat{rgt}:=\frac{1}{n}\sum_{i=1}^n\widehat{rgt_i}$. 
We obtain the empirical regret for each bidder by executing gradient ascent on her bids $b_i$ for $200$ iterations.
We run such gradient ascent for $100$ times with different initial bids $b_i^{(0)}$, and the maximum regret is recorded for bidder $i$.

\begin{table*}[t]
    \centering
    \caption{Experiment results of known settings (Setting \ref{setting:3(5)x1}-\ref{setting:5x1_10d}).
    The optimal solutions are given by~\citet{myerson1981optimal}.
    Each experiment is run $5$ times and the average results are presented.}
    % \resizebox{\textwidth}{!}{
    \begin{tabularx}{0.8\textwidth}{lcccccc}
        \toprule
        \multirow{3}{*}{Method} 
        & \multicolumn{2}{c}{\ref{setting:3(5)x1}: $3\times 1$} 
        & \multicolumn{2}{c}{\ref{setting:3(5)x1(2)}: $3\times 1$} 
        & \multicolumn{2}{c}{\ref{setting:5x1_10d}: $5 \times 1$}
        \\
        & \multicolumn{2}{c}{$|\Xx|=5,|\Yy|=1$} 
        & \multicolumn{2}{c}{$|\Xx|=5,|\Yy|=2$} 
        & \multicolumn{2}{c}{$\Xx, \Yy \subset \RR^{10}$}
        \\
        & $rev$ & $rgt$ & $rev$ & $rgt$ & $rev$ & $rgt$ \\
        \midrule \midrule
        $\mathtt{Optimal}$ & 0.594 & - & 0.456 & - & 0.367 & -\\ 
        \midrule
        $\mathtt{RegretNet}$ & 0.516 & $<$0.001 & 0.412 & $<$0.001 & 0.329 & $<$0.001\\
        $\mathtt{EquivariantNet}$ & 0.498& $<$0.001 & 0.403 & $<$0.001 & 0.311& $<$0.001\\
        \midrule
        \mtt{CIRegretNet} & 0.594 & $<$0.001 & 0.453 & $<$0.001 & 0.364 & $<$0.001 \\
        \mtt{CIEquivariantNet}& 0.590& $<$0.001 & 0.452& $<$0.001 & 0.360& $<$0.001
        \\
        \midrule
        \name 
        & 0.593 & $<$0.001 
        & 0.454 & $<$0.001 
        & 0.366 & $<$0.001
        \\
        \bottomrule
    \end{tabularx}
    % }
    \label{tab:single_setting}
\end{table*}

\paragraph{Single-item Contextual Auctions}
First, we evaluate \name~in single-item auctions, whose optimal solutions are given by \citet{myerson1981optimal}. We aim to justify whether \name~can recover the near-optimal solutions.
The specific single-item auctions we consider are:
\begin{enumerate} [label=(\Alph*),ref=\Alph*]
% ,itemsep=2pt,topsep=0pt,parsep=0pt]
    \item \label{setting:3(5)x1} 
    $3$ bidders and $1$ item, with discrete bidder-contexts and item-context, in which $\Xx=\{1, 2, 3, 4, 5\}$ and $\Yy=\{1\}$.
    Both contexts are independently and uniformly sampled.
    Given $x_i \in \Xx$ and $y_1=1$, $v_{i1}$ is drawn according to the truncated normal distribution $\Nn(\frac{x_i}{6}, 0.1)$ in $[0, 1]$.
    % \label{setting:1}
    
    \item\label{setting:3(5)x1(2)} 
    $3$ bidders and $1$ item, with discrete bidder-contexts and item-context, in which $\Xx=\{1, 2, 3, 4, 5\}$ and $\Yy=\{1, 2\}$.
    Both contexts are independently and uniformly sampled.
    Given $x_i \in \Xx$, $v_{i1}$ is drawn according to the truncated normal distribution $\Nn(\frac{x_i}{6}, 0.1)$ in $[0, 1]$ when $y_1 = 1$, and is drawn according to probability densities $f_i(x) = \frac{i}{6}e^{-\frac{i}{6}x}$ truncated in $[0, 1]$ when $y_1 = 2$.
    
    \item\label{setting:5x1_10d} $5$ bidders and $1$ item, with continuous bidder-contexts and item-context, in which $\Xx = [-1, 1]^{10}$ and $\Yy = [-1, 1]^{10}$. 
    Both the contexts are independently and uniformly sampled.
    Given $x_i \in \Xx$ and $y_j \in 
    \Yy$, $v_{ij}$ is drawn according to $U[0, \mathrm{Sigmoid}(x_i^Ty_j)]$.
\end{enumerate}

We present the experimental results of Setting \ref{setting:3(5)x1}, \ref{setting:3(5)x1(2)} and \ref{setting:5x1_10d} in \cref{tab:single_setting}. 
We can see that all the context-integrated models (\mtt{CIRegretNet, CIEquivariantNet} and \name) are able to recover the optimal solutions given by~\citet{myerson1981optimal} in these simple settings: near-optimal revenues are achieved with regrets less than $0.001$.
In comparison, despite low regret, \mtt{RegretNet} and \mtt{EquivariantNet} fail to reach the optimal solution. 
It turns out that integrating context information into model architecture is crucial in contextual auction design.
Furthermore, \mtt{EquivariantNet}, the symmetric mechanism designer, fails to reach the same performance as \mtt{RegretNet}, which reflects the importance of designing asymmetric solutions in contextual auctions.

\begin{table*}[t]
    \centering
    \caption{Experiment results for Setting \ref{setting:2(10)x5(10)}-\ref{setting:5x10_10d}. Each experiment is run by $5$ times and the 
    average results are presented.}
    
    \resizebox{\textwidth}{!}{
    \begin{tabularx}{1.3\textwidth}{lcccccccccccc}
        \toprule
        
        \multirow{3}{*}{Method} 
        & \multicolumn{2}{c}{ \ref{setting:2(10)x5(10)}: $2\times 5$} 
        & \multicolumn{2}{c}{ 
        \ref{setting:3(10)x10(10)}: $3\times 10$}
        & \multicolumn{2}{c}{ \ref{setting:5(10)x10(10)}: $5 \times 10$}
        & \multicolumn{2}{c}{ \ref{setting:2x5_10d}: $2\times 5$} 
        & \multicolumn{2}{c}{ 
        \ref{setting:3x10_10d}: $3\times 10$}
        & \multicolumn{2}{c}{ \ref{setting:5x10_10d}: $5 \times 10$}
        \\
        & \multicolumn{2}{c}{ $|\Xx|=|\Yy|=10$} 
        & \multicolumn{2}{c}{$|\Xx|=|\Yy|=10$} 
        & \multicolumn{2}{c}{$|\Xx|=|\Yy|=10$}
        & \multicolumn{2}{c}{ $\Xx,\Yy \subset \RR^{10}$} 
        & \multicolumn{2}{c}{ 
        $\Xx,\Yy \subset \RR^{10}$} 
        & \multicolumn{2}{c}{ $\Xx,\Yy \subset \RR^{10}$}
        \\
        & $rev$ & $rgt$ & $rev$ & $rgt$ & $rev$ & $rgt$ & $rev$ & $rgt$ & $rev$ & $rgt$ & $rev$ & $rgt$\\
        \midrule \midrule
        \mtt{Item\text{-}wise\;Myerson} & 2.821 & - & 6.509 & - & 7.376 & - & 1.071 & - & 2.793 & - & 3.684 & - \\
        \midrule
        \mtt{CIRegretNet} & 2.803 & $<$0.001 & 5.846 & $<$0.001 & 6.339 & $<$0.003 & 1.104 & $<$0.001 & 2.424 & $<$0.001 & 2.999 & $<$0.001 \\
        \mtt{CIEquivariantNet} & 2.841& $<$0.001 & 6.703& $<$0.001&7.602 & $<$0.003 &1.147 & $<$0.001& 2.872& $<$0.001&3.806 & $<$0.001 \\
        \midrule
        \name 
        & \textbf{2.916} & $<$0.001 
        & \textbf{6.872} & $<$0.001 
        & \textbf{7.778} & $<$0.003
        & \textbf{1.177} & $<$0.001 
        & \textbf{2.918} & $<$0.001 
        & \textbf{3.899} & $<$0.001 
        \\
        \bottomrule 
    \end{tabularx}
    }
    
    \label{tab:complex_setting}
\end{table*}
\paragraph{Multi-item Contextual Auctions}
Next, we illustrate the potential of \name~to discover new auction designs in multi-item contextual auctions without known solutions.
We consider discrete context settings as follows:
\begin{enumerate}
[label=(\Alph*),start=4,ref=\Alph*]
% itemsep=2pt,topsep=0pt,parsep=0pt]
    \item\label{setting:2(10)x5(10)} 
    $2$ bidders with $\Xx=\{1, 2, \dots, 10\}$ and $5$ items with $\Yy=\{1, 2, \dots, 10\}$.
    All the contexts are uniform sampled, and $v_{ij}$ is drawn according to the normal distribution $\Nn\left(\frac{(x_i+y_j)\bmod 10 + 1}{11}, 0.05\right)$ truncated in $[0, 1]$.
    \item
    \label{setting:3(10)x10(10)} $3$ bidders and $10$ items. The discrete contexts and corresponding values are drawn similarly as Setting \ref{setting:2(10)x5(10)}.
    \item
    \label{setting:5(10)x10(10)} $5$ bidders and $10$ items, which is, to the best of our knowledge, the largest auction size considered in previous literatures of deep learning based auction design~\citep{rahme2020auction}. The discrete contexts and corresponding values are drawn similarly as Setting \ref{setting:2(10)x5(10)}.
\end{enumerate}
Additionally, We also construct continuous context settings based on Setting \ref{setting:5x1_10d}:
\begin{enumerate}
[label=(\Alph*),start=7,ref=\Alph*]
% ,itemsep=2pt,topsep=0pt,parsep=0pt]
\item\label{setting:2x5_10d} 
$2$ bidders and $5$ items. The continuous contexts and corresponding values are drawn similarly as Setting \ref{setting:5x1_10d}.

\item\label{setting:3x10_10d} 
$3$ bidders and $10$ items. The continuous contexts and corresponding values are drawn similarly as Setting \ref{setting:5x1_10d}.

\item\label{setting:5x10_10d} 
$5$ bidders and $10$ items. The continuous contexts and corresponding values are drawn similarly as Setting \ref{setting:5x1_10d}.
\end{enumerate}
Experimental results for Setting \ref{setting:2(10)x5(10)}-\ref{setting:5x10_10d} are shown in \cref{tab:complex_setting}.
\name~obtains the best revenue results in all the settings while keeping low regret (less than $0.003$ in Setting \ref{setting:5(10)x10(10)} and less than $0.001$ in all the other settings).
Notice that the only difference between \name, \mtt{CIRegretNet} and \mtt{CIEquivariantNet} is the architecture of interaction layers.
Such a result indicates the effectiveness of our transformer-based interaction module to capture the complex mutual influence among bidders and items. 
Furthermore, both \name~and~\mtt{CIEquivariantNet} outperform \mtt{CIRegretNet} a lot in all the $3\times 10$ and $5\times 10$ auctions, showing that adding the inductive bias of permutation-equivariance is helpful in large-scale auction design.

\pgfplotsset{width=1.09\textwidth,compat=newest}
\begin{figure*}[t]
\centering 

\begin{subfigure}[b]{0.33\textwidth}
\centering
\begin{tikzpicture}[font=\footnotesize] 
\begin{axis}[
    xlabel=Number of Bidders, 
    ylabel=Revenue, 
    tick align=inside, 
    xtick={3,4,5,6,7},
    legend pos=south east,
    ymajorgrids=true,
    xmajorgrids=true,
    grid style=dashed,
]
    \addplot[smooth,mark=square*,blue] plot coordinates { 
        (3,6.872)
        (4,7.222)
        (5,7.395)
        (6,7.496)
        (7,7.598)
    };
    \addlegendentry{\name}
    \addplot[smooth,mark=triangle*,red] plot coordinates {
        (3,6.509)
        (4,7.028)
        (5,7.376)
        (6,7.629)
        (7,7.837)
    };
    \addlegendentry{\mtt{Baseline}}

\end{axis}
\end{tikzpicture}
\subcaption{
% Trained on Setting \ref{setting:3(10)x10(10)}.
}
\label{subfig:oos_1}
\end{subfigure}%
\begin{subfigure}[b]{0.33\textwidth}
\centering
\begin{tikzpicture}[font=\footnotesize] 
\begin{axis}[
    xlabel=Number of Items, 
    ylabel=Revenue, 
    tick align=inside, 
    xtick={3,4,5,6,7},
    legend pos=south east,
    ymajorgrids=true,
    xmajorgrids=true,
    grid style=dashed,
]
    \addplot[smooth,mark=square*,blue] plot coordinates { 
        (3,1.720)
        (4,2.333)
        (5,2.916)
        (6,3.540)
        (7,4.141)
    };

    \addlegendentry{\name}
    \addplot[smooth,mark=triangle*,red] plot coordinates {
        (3,1.691)
        (4,2.264)
        (5,2.821)
        (6,3.391)
        (7,3.954)
    };
    \addlegendentry{\mtt{Baseline}}
\end{axis}
\end{tikzpicture}
\subcaption{
% Trained on Setting \ref{setting:2(10)x5(10)}.
}
\label{subfig:oos_2}
\end{subfigure}%
\begin{subfigure}[b]{0.33\textwidth}
\centering
\begin{tikzpicture}[font=\footnotesize] 
\begin{axis}[
    xlabel=Number of Items, 
    ylabel=Revenue, 
    tick align=inside, 
    xtick={3,4,5,6,7},
    legend pos=south east,
    ymajorgrids=true,
    xmajorgrids=true,
    grid style=dashed,
]
    \addplot[smooth,mark=square*,blue] plot coordinates { 
        (3,0.677)
        (4,0.919)
        (5,1.177)
        (6,1.438)
        (7,1.686)
    };
    \addlegendentry{\name}
    \addplot[smooth,mark=triangle*,red] plot coordinates {
        (3,0.640)
        (4,0.855)
        (5,1.071)
        (6,1.290)
        (7,1.489)
    };
    \addlegendentry{\mtt{Baseline}}
\end{axis}
\end{tikzpicture}
\subcaption{
% Trained on Setting \ref{setting:2x5_10d}.
}
\label{subfig:oos_3}
\end{subfigure}
\caption{
Out-of-setting generalization results: we train {\name}  and evaluate it on the same contextual auction with a different number of bidders or items.
We set \mtt{Item\text{-}Wise~Myerson} as the baseline.
The regret results are less than $0.001$ in all of these experiments.
(a) Trained on Setting  \ref{setting:3(10)x10(10)} ($3\times 10$ with $|\Xx|=|\Yy|=10$) and evaluated with different number of bidders. 
(b) Trained on Setting \ref{setting:2(10)x5(10)} ($2\times 5$ with $|\Xx|=|\Yy|=10$) and evaluated with different number of items.
(c) Trained on Setting \ref{setting:2x5_10d} ($2\times 5$ with $\Xx,\Yy \subset \RR^{10}$) and evaluated with different number of items.
}
\label{fig:oos}
\end{figure*}
\paragraph{Out-of-setting Generalization}
In addition, to show the effectiveness of \name, we also conduct out-of-setting generalization experiments.
Specifically, we train our model and evaluate it in auctions with a different number of bidders or items. 
Such evaluation is feasible for \name, since the size of parameters in \name~does not rely on the number of bidders and items. 
We illustrate the experimental results on \cref{fig:oos}, and see \cref{app:exp:oos} for more detailed numerical values. 
\cref{subfig:oos_1} shows the experimental results of generalizing to a varying number of bidders.
We train \name~on Setting \ref{setting:3(10)x10(10)}, the discrete context settings with $3$ bidders and $10$ items, and we evaluate \name~on the same contextual auction with $n$ bidders and $10$ items ($n \in \{3,4,5,6,7\}$).
We observe good generalization results: In addition to obtain low regret (less than $0.001$) in all the test settings, \name~outperforms \mtt{Item\text{-}wise~Myerson} when $n\in\{3,4,5\}$
\footnote{As comparison, we find \mtt{CIEquivariantNet} fails to generalize to different bidders. See \cref{app:exp:oos} for the results.}.
Furthermore, in \cref{subfig:oos_2} and \cref{subfig:oos_3} we present the experimental results of generalizing to varying number of items.
We train \name~on Setting \ref{setting:2(10)x5(10)} and Setting \ref{setting:2x5_10d} respectively, where both settings have $2$ bidders and $5$ items , and we test the model on the same contextual auction with $2$ bidders and $m$ items ($m \in \{3,4,5,6,7\}$).
Again, we observe good generalization results.
While still keeping small regret (less than $0.001$), \name~is able to outperform \mtt{Item\text{-}wise~Myerson} in all the test auctions.

\section{Conclusion}
\label{sec:conclusion}
In this paper, we propose a new (transformer-based) neural architecture, $\mathtt{CITransNet}$, for contextual auction design. \name~is permutation-equivariant with respect to bids and contexts, and it can handle asymmetric information in auctions.
We show by experiments that $\mathtt{CITransNet}$ can recover the known optimal analytical solutions in simple auctions, and we demonstrate the effectiveness of the transformer-based interaction layers in \name~by comparing \name~with the context integrated version of \mtt{RegretNet} and \mtt{EquivariantNet}.
Furthermore, we also illustrate the out-of-setting generalization ability for \name~by evaluating it in auctions with a varying number of bidders or items. 
Given the decent generalizability of \name, an immediate next step is to test \name~over an industry-scale dataset. 
It would also be interesting to test \name~in an online manner.

\section*{Acknowledgements}

This work is supported by Science and Technology Innovation
2030 - ``New Generation Artificial Intelligence'' Major Project (No.
2018AAA0100901).
We thank Aranyak Mehta and Di Wang for an insightful discussion on the initial version of this paper. We thank all anonymous reviewers for their helpful feedback.

\bibliographystyle{plainnat}
\bibliography{reference}

\clearpage
\onecolumn
\appendix

\section{Transformer Architecture}
\label{app:transformer}
Transformer architecture~\cite{vaswani2017attention}  aims at modeling the mutual correlations among a set of tokens (e.g., words in a sentence in machine translation) via multi-head self-attention module. 
In our paper, we use transformer to model the interactions among the items (or bidders) with respect to a fixed bidder (or item).
Without loss of generality, we denote the input as
\begin{equation*}
    E_{\mathrm{input}} = (e_1, e_2, \dots, e_n)^T \in \RR^{n\times d},
\end{equation*} 
where $n$ is the number of tokens (i.e., bidders or items) and $d$ is the dimension for each feature vector $e_i$.

Let $d_h$ be the hidden dimension of transformer, and $H$ be the number of heads (i.e., subspace).
For head $h \in [H]$, we use the key-value attention mechanism~\cite{miller2016key} to determine which feature combinations are meaningful in the corresponding subspace. 
Specifically, for each token $i \in [n]$, we first compute the correlation between token $i$ and token $j$ in head $h$:
\begin{equation*}
    \alpha_{i,j}^{(h)} = \frac{\exp(\psi^{(h)}({e_i}, {e_j}))}{\sum_{k=1}^{n}\exp(\psi^{(h)}({e_i}, {e_k}))},
\end{equation*}
where
\begin{equation*}
    \psi^{(h)}({e_i}, {e_j}) = \langle {W^{(h)}_{\mathrm{query}}}{e_i}, {W^{(h)}_{\mathrm{key}}}{e_j}  \rangle,
\end{equation*}
is an attention function which defines the similarity between the token $i$ and $j$ under head $h$.
$\langle \cdot, \cdot \rangle$ is inner product, and ${W^{(h)}_\mathrm{query}}$, ${W^{(h)}_\mathrm{key}}\in \mathbb{R}^{d'\times d}$ are transformation matrices which map the original embedding space $\mathbb{R}^{d}$ into a $d' = \frac{d_h}{H}$ dimensional space $\mathbb{R}^{d'}$. 

Next, we update the representation of token $i$ in subspace $h$ by combining all relevant features. This is done by computing the weighted sum using coefficients ${\alpha_{i,j}^{(h)}}$:
\begin{equation*}
    {\widetilde{e}_i^{(h)}} =\sum_{j=1}^{n}{\alpha_{i,j}^{(h)}}({W_\mathrm{value}^{(h)}}{e_j}) \in \RR^{d'},
\end{equation*}
where ${W_\mathrm{value}^{(h)}} \in \mathbb{R}^{d'\times d}$.
Since ${\widetilde{e}_i^{(h)}}\in \mathbb{R}^{d'}$ is a combination of token $i$ and all its relevant tokens, it represents a new combinatorial feature. 

Afterwards, we collect combinatorial features learned in all subspaces as follows:
\begin{equation*}
    %\widetilde{e}_i = \mathrel{\mathop{\Vert}\limits_{h=1}^{H}}\widetilde{e}_i^{(h)}
    {\widetilde{e}_i} = {\widetilde{e}_i^{(1)}} \oplus{\widetilde{e}_i^{(2)}}\oplus\dots\oplus{\widetilde{e}_i^{(H)}} \in \RR^{Hd'} = \RR^{d_h},
\end{equation*}
where $\oplus$ is the concatenation operator, and $H$ is the number of total heads. 

Finally, a token-wise MLP is applied to each token $i$ and we get a new representation for it. 
\begin{equation*}
    e_i' = \mathrm{MLP}(\widetilde{e}_i) \in \RR^{d_h},
\end{equation*}
and the final output is 
\begin{equation*}
    E_{\mathrm{output}} = (e_1', e_2', \dots, e_n')^T \in \RR^{n\times d_h}.
\end{equation*}

Notice that the parameters to be optimized in transformer are $W^{(h)}_{\mathrm{query}}, W^{(h)}_{\mathrm{key}}, W^{(h)}_{\mathrm{value}} \in \RR^{d'\times d}$ for all $h \in [H]$ and the parameters of the final token-wise MLP, all of which are unrelated to the number of tokens $n$.
Furthermore, the transformer architecture is permutation-equivariant.

\section{Optimization and Training Procedures}
\label{app:algorithm}
We use the augmented Lagrangian method to solve the constrained training problem in \cref{eq:empirical_prob} over the space of neural network parameters $w \in \RR^{d_w}$.  
We define the Lagrangian function for the optimization problem augmented with a quadratic penalty term for violating the constraints as mentioned in \cref{eq:loss}.
\begin{equation*}
\begin{aligned}
    \mathcal{L}_{\rho}(w;\lambda) = &-\frac{1}{L}\sum_{\ell=1}^L \sum_{i=1}^n p_i^w(V^{(\ell)}, {x}^{(\ell)}, {y}^{(\ell)})+
    &\sum_{i=1}^n\lambda_i\widehat{rgt}_i(w)+\frac{\rho}{2}\sum_{i=1}^n\left(\widehat{rgt}_i(w)\right)^2
\end{aligned}
\end{equation*}

\begin{algorithm}[t]
\caption{\name~Training}
\label{alg:training}
\begin{algorithmic}[1]
\STATE {\bfseries Input:} Minibatches $\mathcal{S}_1, \ldots, \mathcal{S}_T$ of size $B$%
\STATE {\bfseries Parameters:} 
$\forall t \in [T], \rho_t>0$, $\gamma>0$, $\eta>0$, $c>0$, $T \in \NN$, $\Gamma \in \NN$, $T_\lambda \in \NN$
\STATE {\bfseries Initialize:} $w^0 \in \RR^d$, $\lambda^0 \in \RR^n$ %
\FOR{$t~=~0$ \textbf{to} $T$}
\STATE Receive minibatch $\mathcal{S}_t \,=\, \{(v^{(1)}, x^{(1)}, y^{(1)}), \ldots, (v^{(B)}, x^{(B)}, y^{(B)})\}$
\STATE Initialize misreports ${v'}_{i}^{(\ell)}\in \Vv_i, \forall \ell \in [B],~ i \in N$
\FOR{$r~=~0$ \textbf{to} $\Gamma$}
\STATE ~~$\forall \ell \in [B],~ i \in N:$
\STATE ~~~~\quad${v'}_{i}^{(\ell)} \leftarrow {v'}_{i}^{(\ell)} +\gamma\nabla_{v'_i}\,u^w_i\big(v^{(\ell)}_i, \big({v'}_{i}^{(\ell)}, v^{(\ell)}_{-i}\big),{{x^{(\ell)}}},{{y^{(\ell)}}}\big)$
\ENDFOR
\STATE Compute regret gradient: 
\STATE ~~\quad$\forall \ell \in [B], i \in N$:
\STATE ~~\quad~~\quad$g^t_{\ell, i} = \nabla_w\left[u^w_i(v^{(\ell)}_i,({v'}_{i}^{(\ell)},v^{(\ell)}_{-i}),{{x^{(\ell)}}},{{y^{(\ell)}}})-u^w_i(v^{(\ell)}_i, v^{(\ell)},{{x^{(\ell)}}},{{y^{(\ell)}}})\right]\Big\vert_{w=w^t}$
\STATE Compute Lagrangian gradient using \cref{eq:L-grad} and update $w^t$:
\STATE ~~\quad$w^{t+1} \leftarrow w^t \,-\, \eta\nabla_w\, \Ll_{\rho_t}	(w^{t}, \lambda^{t})$

\STATE Update Lagrange multipliers $\lambda$ once in $T_\lambda$ iterations:
\STATE ~~\quad\textbf{if} {$t$ is a multiple of $T_\lambda$} \textbf{then}
\vskip 2pt
\STATE ~~\quad~~\quad$\lambda^{t+1}_i \leftarrow \lambda_i^{t} + \rho_t\,\widehat{\mathit{rgt}}_i(w^{t+1}), ~~\forall i \in N$
\STATE ~~\quad\textbf{else}  
\STATE ~~\quad~~\quad$\lambda^{t+1} \leftarrow \lambda^t$
\ENDFOR
\end{algorithmic}
\end{algorithm}

\cref{alg:training} describe the training procedure of \name.
First, for each iteration $t \in [T]$, we randomly draw a minibatch $\Ss_t$ of size $B$, in which $\mathcal{S}_t \,=\, \{(v^{(1)}, x^{(1)}, y^{(1)}), \ldots, (v^{(B)}, x^{(B)}, y^{(B)})\}$.
Afterward, we alternately update the model parameters and the Lagrange multipliers:
\begin{enumerate}[label=(\alph*)]
    \item $w^{new}\in\arg\min_w \mathcal{L}_{\rho}(w^{old},\lambda^{old})$
    \item $\lambda_i^{new}=\lambda_i^{old}+\rho\cdot \widehat{rgt}_i(w^{new}), \forall i \in N$
    % \item $\rho^{new}=\rho^{old}+c$
\end{enumerate}

The update (a) is performed approximately using gradient descent.
% Let $\widehat{\mathit{rgt}}_i(w)$ denote the empirical regret in \ref{eq:emp_reg} computed on minibatch $\mathcal{S}_t$.
The gradient of $\Ll_\rho$ w.r.t.\ $w$ for fixed $\lambda^t$ is given by:
\begin{equation}
\label{eq:L-grad}
\nabla_w \, \Ll_\rho(w, \lambda^{t}) = -\frac{1}{B}\sum_{\ell=1}^B \sum_{i\in N} \nabla_w\, p^w_i(v^{(\ell)},x^{(\ell)},y^{(\ell)})
+\, \sum_{i\in N}\, \sum_{\ell = 1}^B \lambda^t_{i}\, g_{\ell, i}
  \,+\,\rho \sum_{i\in N} \, \sum_{\ell = 1}^B\, \widehat{\mathit{rgt}}_i(w)\, g_{\ell, i},
\end{equation}

where
\begin{align*}
g_{\ell, i} ~=~ \nabla_w\Big[\max_{{v'}_{i}^{(\ell)}\in \Vv_i} u^w_i(v^{(\ell)}_i, ({v'}_{i}^{(\ell)},v^{(\ell)}_{-i}),{{x^{(\ell)}}},{{y^{(\ell)}}})-u^w_i(v^{(\ell)}_i, v^{(\ell)},{{x^{(\ell)}}},{{y^{(\ell)}}})\Big].
\end{align*}

The computation of $\widehat{rgt}_i$ and $g_{\ell, i}$ involve a ``max'' over misreports for each bidder $i$, and we solve it approximately by gradient ascent.
In particular, we maintain misreports ${v'}_{i}^{(\ell)}$ for each bidder $i$ on each sample $\ell$. 
For every update on the model parameters $w^t$, we perform $\Gamma$ gradient ascent updates to compute the optimal misreports.

\section{Implementation Details}
\label{app:implement}
For all the settings (Setting \ref{setting:3(5)x1}-\ref{setting:5x10_10d}), we generate each training set with size in $\{50000, 100000, 200000\}$ and test set of size $5000$.

For all the methods, we train the models for a maximum of $80$ epochs with batch size $500$. 
We set the embedding size in settings with discrete context (Setting \ref{setting:3(5)x1}, \ref{setting:3(5)x1(2)}, \ref{setting:2(10)x5(10)}, \ref{setting:3(10)x10(10)}, \ref{setting:5(10)x10(10)}) as $16$.
The value of $\rho$ in the augmented Lagrangian (\cref{eq:loss}) was set as $1.0$ at the beginning and incremented by $5$ every two epochs.
The value of $\lambda$ in \cref{eq:loss} was set as $5.0$ initially and incremented every certain number (selected from $\{2-10\}$) of epochs.
All the models and regret are optimized through Adam~\citep{kingma2014adam} optimizer.
Following \citet{dutting2019optimal}, for each update on model parameters, we run $\Gamma=25$ update steps on the misreport bid $b_i$ for each bidder, and the optimized misreports are cached to initialize the misreports bidding in the next epoch.

For our proposed \name, the output channel of the first $1\times 1$ convolution in both the input layer and interaction layers are set to $64$.
We set $d=64$ for the $1\times 1$ convolution with residual connection in input layer, and $d_h=64$ for the final $1\times 1$ convolution in each interaction layer.
We tune the numbers of interaction layers from $\{2, 3\}$, and in each interaction layer we adopt transformer with $4$ heads and $64$ hidden nodes.

\mtt{RegretNet} and \mtt{CIRegretNet} take fully-connected neural networks as the core architecture.
We choose the number of layers from $\{3, 4, 5, 6, 7\}$ and the number of hidden nodes per layer from $\{64, 128, 256\}$.
As for \mtt{EquivariantNet} and \mtt{CIEquivariantNet}, we use $4$ exchangeable matrix layers of $64$ channels each. 

\section{Additional Out-of-Setting Generalization Experiments}
\label{app:exp:oos}
\begin{table*}[t]
\label{subtab:generalization:n}
    \centering
\caption{
Out-of-setting generalization results of \name and \mtt{CIEquivariantNet}: we train each model and evaluate it on the same contextual auction with a different number of bidders or items.
(a) Trained on Setting  \ref{setting:3(10)x10(10)} ($3\times 10$ with $|\Xx|=|\Yy|=10$) and evaluated with different number of bidders. 
(b) Trained on Setting \ref{setting:2(10)x5(10)} ($2\times 5$ with $|\Xx|=|\Yy|=10$) or Setting \ref{setting:2x5_10d} ($2\times 5$ with $|\Xx|\subset \RR^{10},|\Yy|\subset \RR^{10}$) and evaluated with different number of items.
}
    
\begin{subtable}[h]{\textwidth}
\label{subtab:generalization:n1}
    \centering
    \subcaption{}
    \resizebox{\textwidth}{!}{
    \begin{tabularx}{1.15\textwidth}{lcccccccccc}
        \toprule
        \multirow{2}{*}{Method} 
        & \multicolumn{2}{c}{$3\times 10$}
        & \multicolumn{2}{c}{$4\times 10$} 
        & \multicolumn{2}{c}{$5\times 10$}
        & \multicolumn{2}{c}{$6\times 10$}
        & \multicolumn{2}{c}{$7\times 10$}
        \\
        & $rev$ & $rgt$ & $rev$ & $rgt$ & $rev$ & $rgt$ & $rev$ & $rgt$ & $rev$ & $rgt$ \\
        \midrule \midrule
        \multicolumn{5}{l}{Trained on Setting \ref{setting:3(10)x10(10)}: $n=3, m=10$ with $|\Xx| = |\Yy| = 10$}
        \\
        \midrule
        \mtt{Item\text{-}wise\;Myerson} & 
        6.509 & - & 7.028 & - & 7.376 & - & \textbf{7.629} & - & \textbf{7.837} & -
        \\
        \mtt{CIEquivariantNet} & 
        6.703 & $<$0.001 & 7.024 & 0.018 & 7.229 & 0.051 & 7.365 & 0.079& 7.474&0.1
        \\
        {\name} & 
        \textbf{6.872} & $<$0.001 & \textbf{7.222} & $<$0.001 & \textbf{7.395} & $<$0.001 & 7.496 & $<$0.001 & 7.598 &  $<$0.001
        \\
        \bottomrule 
    \end{tabularx}
    }
\end{subtable}

\vskip 0.1in

\begin{subtable}[h]{\textwidth}
\label{subtab:generalization:n2}
    \centering
    \subcaption{}
    \resizebox{\textwidth}{!}{
    \begin{tabularx}{1.15\textwidth}{lcccccccccc}
        \toprule
        \multirow{2}{*}{Method} 
        & \multicolumn{2}{c}{$2\times 3$} 
        & \multicolumn{2}{c}{$2\times 4$}
        & \multicolumn{2}{c}{$2\times 5$}
        & \multicolumn{2}{c}{$2\times 6$} 
        & \multicolumn{2}{c}{$2\times 7$}
        \\
        & $rev$ & $rgt$ & $rev$ & $rgt$ & $rev$ & $rgt$ & $rev$ & $rgt$ & $rev$ & $rgt$ \\
        \midrule \midrule
        \multicolumn{11}{l}{Trained on Setting \ref{setting:2(10)x5(10)}: $n=2, m=5$ with $|\Xx| = |\Yy| = 10$}
        \\
        \midrule
        \mtt{Item\text{-}wise\;Myerson} & 1.691 & - & 2.264 & - & 2.821 & - & 3.391 & - & 3.954 & -
        \\
        \mtt{CIEquivariantNet} & 
        1.687 & $<$0.001 & 2.267 & $<$0.001 & 2.841 & $<$0.001 & 3.405 & $<$0.001& 3.971&$<$0.001
        \\
        {\name} & \textbf{1.720} & $<$0.001 & \textbf{2.333} & $<$0.001 & \textbf{2.916} & $<$0.001 & \textbf{3.540} & $<$0.001 & \textbf{4.141} & $<$0.001 
        \\
        \midrule
        \multicolumn{11}{l}{Trained on Setting \ref{setting:2x5_10d}: $n=2, m=5$ with $\Xx,\Yy\subset \RR^{10}$}
        \\
        \midrule
        \mtt{Item\text{-}wise\;Myerson} & {0.640} & - & 0.855 & - & 1.071 & - & 1.290 & - & 1.489 & -
        \\
        \mtt{CIEquivariantNet} & 
        0.663 & $<$0.001 & 0.900 & $<$0.001 & 1.147 & $<$0.001 & 1.400 & $<$0.001& 1.637&$<$0.001
        \\
        {\name} & \textbf{0.677} & $<$0.001 & \textbf{0.919} & $<$0.001 & \textbf{1.177} & $<$0.001 & \textbf{1.438} & $<$0.001 & \textbf{1.686} & $<$0.001 
        \\
        \bottomrule 
    \end{tabularx}
    }
\end{subtable}

    \label{tab:generalization}
\end{table*}

In addition to \name, we also conduct the same out-of-setting generalization experiments for \mtt{CIEquivariantNet}. The numerical results are shown in \cref{tab:generalization}. 
While \name~generalize well to all of these settings with low regret (less than $0.001$), \mtt{CIEquivariantNet} fails to obtain low regret when generalizing to auctions with a different number of bidders. 
Such a result indicates the critical role of the transformer-based interaction layers in \name~when generalized to settings with varying bidders.

\section{Proof of \cref{thm:UC}}
\label{app:proof:thm:UC}
% \subsection{Proof Scratch and Definitions}

The proof is done by combining covering numbers~\citep{shalev2014understanding, dutting2019optimal} and a generalization Lemma (\cref{lemma:generalization}, whose technique comes from \citet{duan2021pac}) based on concentration inequality.

\subsection{Basic Definition}
On top of the definitions in \cref{subsec:theory}, we first define the covering numbers of bidder's utility functions and regret functions. 

\paragraph{Covering Numbers $\Nn_{\infty, 1}(\Uu, r)$ and $\Nn_{\infty}(\Uu_i, r)$} Let $\mathcal{U}_i$ be the class of utility functions for bidder $i$ on auctions in $\Mm$, i.e.,
\begin{equation*}
\mathcal{U}_i =
\Big\{
 u_i: \Vv_i \times \Vv \times\Xx^n\times {\Yy}^m\rightarrow \RR \,\Big|\, 
 u_i({v}_i, v,{{x}},{{y}}) \,=\, \sum_{j=1}^m g_{ij}(v,{{x}},{{y}}) v_{ij} - p_i(v,{{x}},{{y}})
\Big\}.
\end{equation*}
Similarly, let $\mathcal{U}$ be the class of utility profiles over $\Mm$. 
Define the $\ell_{\infty, 1}$-distance between two utility profiles $u$ and $u'$ as $ \max_{v, v',x,y} \sum_{i=1}^n \vert u_i(v_i,(v'_i,v_{-i}),x,y) -u_i(v_i,(v'_i,v_{-i}),x,y)\vert$
and $\mathcal{N}_{\infty, 1}(\mathcal{U}, r)$ as the minimum number of balls of radius $r>0$ to cover $\mathcal{U}$ ($r$-covering number of $\Uu$) under such $\ell_{\infty, 1}$-distance.
We also define the $\ell_{\infty}$-distance between $u_i$ and $u'_i$ as $\max_{v, v'_i}\vert u_i(v_i,(v'_i, v_{-i}),x,y) - u'_i(v_i,(v'_i, v_{-i})x,y)\vert$ and $\mathcal{N}_\infty(\mathcal{U}_i, r)$ as the $r$-covering number of $\Uu_i$ under $\ell_{\infty}$-distance.

\paragraph{Covering Numbers $\Nn_{\infty, 1}(\mathrm{RGT}, r)$ and $\Nn_{\infty}(\mathrm{RGT}_i, r)$} As for regret functions, let $\mathrm{RGT}_i \circ \mathcal{U}_i$ be the class of all regret functions for bidder $i$, i.e.,
\begin{equation*}
\begin{aligned}
\mathrm{RGT}_i \circ \mathcal{U}_i &=
\Big\{
 {rgt}_i: \Vv \times \Xx^n\times {\Yy}^m\rightarrow \RR \,\Big|
 \\
 \quad&\, rgt_i(v,x,y) \,=\, 
\max_{{v}_i '\in \Vv_i} u_i({v}_i, ({v}_i ',v_{-i}),{{x}},{{y}})-u_i({v}_i, v,{{x}},{{y}})\text{ for some  }
u_i \in \mathcal{U}_i
\Big\}.
\end{aligned}
\end{equation*}
The same as before, we define $\mathrm{RGT}\circ \mathcal{U}$ as 
the class of profiles of regret functions, and we define $\ell_{\infty, 1}$-distance between two regret profiles $rgt$ and $rgt'$ as $\max_{v,x,y} \sum_{i=1}^n \vert rgt_i(v,x,y)-rgt'_i(v,x,y)\vert$.
Let $\Nn_{\infty, 1}(\mathrm{RGT}\circ\mathcal{U}, r)$ denote the $r$-covering number of $\mathrm{RGT}\circ \Uu$ under such distance. 
Similarly, define the $\ell_\infty$-distance between $rgt_i$ and $rgt'_i$ as $\max_{v,x,y}\vert rgt_i(v,x,y)-rgt'_i(v,x,y)\vert$, and denote $\mathcal{N}_\infty(\mathrm{RGT}\,\circ\,\mathcal{U}_i, r)$ as the $r$-covering number of $\mathrm{RGT}$.
 
\paragraph{Covering Numbers $\Nn_{\infty, 1}(\Pp, r)$ and $\Nn_{\infty}(\Pp_i, r)$}  
As for revenue (payment) functions, we denote the class of all the profiles of payment functions as $\mathcal{P}$ and 
\begin{equation*}
    \mathcal{P}_i \,=\, \{p_i: \Vv\times\Xx\times\Yy\rightarrow \RR_{\ge 0}\,|\,p \in  \Pp\}.
\end{equation*}
We denote the $r$-covering number of $\Pp$ as $\mathcal{N}_\infty(\Pp, r)$  under the $\ell_{\infty,1}$-distance and the  $r$-covering number for $\Pp_i$ as $\Nn_\infty(\Pp_i, \epsilon)$ under the $\ell_{\infty}$-distance.

\subsection{Important Lemmas}

The generalization lemma (\cref{lemma:generalization}) plays an important role in our proof.

\begin{lemma}
\label{lemma:generalization}
Let $\mathcal{S} \,=\, \{z_1,\ldots,z_L\} \in \Zz^L$ be a set of samples drawn i.i.d. from some distribution $\Dd$ over $\Zz$. 
We assume $f(z) \in [a, b]$ for all $f\in\mathcal{F}$ and $z\in\Zz$.  
Define the $\ell_\infty$-distance between two functions $f, f'\in \Ff$ as $\max_{z \in \Zz}|f(z) - f'(z)|$ and define $\Nn_{\infty}(\Ff, r)$ as the $r$-covering number of $\Ff$ under such $\ell_\infty$-distance.
Let $\LL_\Dd(f)=\EE_{z\sim D}[f(z)]$ and $\LL_S(f)=\frac{1}{|S|}\sum_{i=1}^{|S|} f(z_i)$, then we have
\begin{equation*}
\PP_{S\sim\Dd^m}\Big[\exists f\in\mathcal{F}, \Big|\LL_S(f)
    - \LL_\Dd(f)]\Big| > \epsilon \Big]\le 2\Nn_{\infty}(\mathcal{F}, \frac{\epsilon}{3})\exp\left(-\frac{2L\epsilon^2}{9(b-a)^2}\right).
\end{equation*}
\end{lemma}
\begin{proof}
Define $\Ff_r$ as the minimum function class that $r$-covers $\Ff$ (so that $|\Ff_r| = \Nn_{\infty}(\Ff, r)$).
For all function $f \in \Ff$, denote $f_r$ as the closed function to $f$ in such function class $\Ff_r$.
On top of that, we have $|f(z) - f_r(z)| \le r, \forall z \in \Zz$. 
For all $\epsilon > 0$, set $r = \frac{\epsilon}{3}$, we get
\begin{equation}
\label{eq:UC}
\begin{aligned}
    &\PP_{S\sim\Dd^m}\Big[\exists f\in\mathcal{F}, \Big|\LL_S(f)
    - \LL_\Dd(f)]\Big| > \epsilon \Big] 
    \\
    \le& \PP_{S\sim\Dd^m}\Big[ \exists f\in\mathcal{F}, \Big|\LL_S(f) - \LL_S({f}_r)\Big| 
     + \Big| \LL_S({f}_r) - \LL_\Dd({f}_r)\Big| 
        + \Big|\LL_\Dd({f}_r) - \LL_\Dd(f)]\Big| 
    > \epsilon \Big] 
    \\
    \le& \PP_{S\sim\Dd^m}\Big[\exists f\in\mathcal{F}, 
        r + \Big| \LL_S({f}_r) - \LL_\Dd({f}_r)\Big| + r > \epsilon \Big] 
    \\
    \le& \PP_{S\sim\Dd^m}\Big[\exists {f}_r \in {\mathcal{F}}_r, \Big|\LL_S({f}_r) - \LL_\Dd({f}_r) \Big| > \frac{1}{3}\epsilon \Big], \quad r = \frac{\epsilon}{3}
    \\
    \le& \Nn_{\infty}(\mathcal{F}, \frac{\epsilon}{3})\PP_{S\sim\Dd^m}\Big[\Big|\LL_S({f}) - \LL_\Dd({f}) \Big| > \frac{1}{3}\epsilon \Big]
    \\
    \overset{(a)}{\le}& 2\Nn_{\infty}(\mathcal{F}, \frac{\epsilon}{3})\exp\left(-\frac{2L\epsilon^2}{9(b-a)^2}\right),
\end{aligned}
\end{equation}
where $(a)$ holds by Hoeffding Inequality.
\end{proof}

The following two lemmas (\cref{cover1} and \cref{cover2}) provides the covering numbers bound for payment and regret.
\begin{lemma}
\label{cover1}
$\mathcal{N}_{\infty,1}(\mathcal{P}, \epsilon) \leq\mathcal{N}_{\infty,1}(\mathcal{M}, \epsilon)$.
\end{lemma}
\begin{proof}
By the definition of the covering number for the auction class $\Mm$, there exists a cover $\hat{\Mm}$ for $\mathcal{M}$ of size $|\hat{\Mm}| \leq  \mathcal{N}_{\infty, 1}(\mathcal{M},\epsilon)$
such that for
 any $(g, p)\in \mathcal{M}$, there is a $(\hat{g}, \hat{p}) \in \hat{\Mm}$  for all $v,x,y$, 
\begin{equation*}
\sum_{i,j} \vert g_{ij}(v,x,y) - \hat{g}_{ij}(v,x,y)\vert + \sum_i\vert p_i(v,x,y) - \hat{p}_i(v,x,y)\vert \leq \epsilon.
\end{equation*}

As a result, we can have $\hat{\Pp}=\{\hat{p} \,\Big|\, (\hat{g},\hat{p})\in \hat M\}$, then for any $p\in\Pp$, there exist a $\hat{p}\in\hat{\Pp}$, for all $v,x,y$,
\begin{equation*}
\sum_i\vert p_i(v,x,y) - \hat{p}_i(v,x,y)\vert \leq \epsilon.
\end{equation*}

Therefore, we have $\mathcal{N}_{\infty, 1}(\mathcal{P}, \epsilon) \leq\mathcal{N}_{\infty, 1}(\mathcal{M}, \epsilon)$.
\end{proof}
\begin{lemma}
\label{cover2}
$\mathcal{N}_{\infty,1}(\mathrm{RGT}\circ\mathcal{U}, \epsilon)\leq\mathcal{N}_{\infty,1}(\mathcal{M}, \frac{\epsilon}{2n})$.
\end{lemma}
\begin{proof}
The proof then proceeds in two steps:
\begin{enumerate}
    \item 
bounding the covering number for each regret class $\mathrm{RGT}\circ\mathcal{U}$ in terms of the covering number for individual utility classes $\mathcal{U}$;

    \item 
bounding the covering number for the joint utility class $\mathcal{U}$ in terms of the covering number for $\Mm$.

\end{enumerate}

First we prove that $\mathcal{N}_{\infty,1}(\mathrm{RGT}\circ\mathcal{U}, \epsilon) \leq \mathcal{N}_{\infty,1}(\mathcal{U}, \frac{\epsilon}{2})$.

By the definition of covering number $\mathcal{N}_{\infty,1}(\mathcal{U}, r)$, 
there exists  a cover $\hat{\Uu}$ with size at most $\Nn_{\infty,1}(\mathcal{U}, \epsilon/2)$ such that
for any $u\in \mathcal{U}$, there is a $\hat{u} \in \hat{\mathcal{U}}$ with
\begin{equation*}
\max_{v, v',x,y} \sum_{i=1}^{n}\Big\vert u_i(v_i, (v'_i, v_{-i}),x,y) - \hat{u}_i(v_i, (v'_i, v_{-i}),x,y)\Big\vert\leq \frac{\epsilon}{2}.
\end{equation*}

For any $u \in \Uu$, taking $\hat{u}\in\hat{\mathcal{U}}$ satisfying the above condition, then for any $v,x,y$, we have
\begin{equation*}
\begin{aligned}
    &\Big\vert\max_{v'_i \in \Vv_i}\big(u_i(v_i, (v'_i, v_{-i}),x,y) - u_i(v_i, (v_i, v_{-i}),x,y)\big) - \max_{\bar{v}_i\in \Vv_i}\big(\hat{u}_i(v_i, (\bar{v}_i, v_{-i}),x,y) - \hat{u}_i(v_i, (v_i, v_{-i}),x,y)\big)\Big\vert\\
\leq~& \Big\vert\max_{v'_i}u_i(v_i, (v'_i, v_{-i}),x,y) - \max_{\bar{v}_i}\hat{u}_i(v_i, (\bar{v}_i, v_{-i}),x,y) + \hat{u}_i(v_i, (v_i, v_{-i}),x,y) - u_i(v_i, (v_i, v_{-i}),x,y)\Big\vert\\
\leq~&\left\vert\max_{v'_i}u_i(v_i, (v'_i, v_{-i}),x,y) - \max_{\bar{v}_i}\hat{u}_i(v_i, (\bar{v}_i, v_{-i}),x,y)\right\vert + \Big\vert \hat{u}_i(v_i, (v_i, v_{-i}),x,y) - u_i(v_i, (v_i, v_{-i}),x,y)\Big\vert\\
\leq~&\left\vert\max_{v'_i}u_i(v_i, (v'_i, v_{-i}),x,y) - \max_{\bar{v}_i}\hat{u}_i(v_i, (\bar{v}_i, v_{-i}),x,y)\right\vert +  \max_{v'_i}\Big\vert u_i(v_i, (v'_i, v_{-i}),x,y) - \hat{u}_i(v_i, (v'_i, v_{-i}),x,y)\Big\vert.
\end{aligned}
\end{equation*}

Let $v^*_i \in \arg\max_{v'_i} u_i(v_i, (v'_i, v_{-i}),x,y)$ and $\hat{v}^*_i \in \arg\max_{\bar{v}_i} \hat{u}_i(v_i, (\bar{v}_i, v_{-i}),x,y)$, then
\begin{equation*}
\begin{aligned}
\max_{v'_i}u_i(v_i, (v'_i, v_{-i}),x,y) -\max_{\bar{v}_i}\hat{u}_i(v_i, (\bar{v}_i, v_{-i}),x,y)
=& u_i(v_i,(v^*_i, v_{-i}),x,y)- \hat{u}_i(v_i, (\hat{v}^*_i, v_{-i}),x,y)
\\
\leq&  u_i(v_i,(v^*_i, v_{-i}),x,y)- \hat{u}_i(v_i, (v^*_i, v_{-i}),x,y)
\\
\leq& \max_{v'_i} \Big\vert u_i(v_i, (v'_i, v_{-i}),x,y) - \hat{u}_i(v_i, (v'_i, v_{-i}),x,y)\Big\vert
\\
\max_{\bar{v}_i}\hat{u}_i(v_i, (\bar{v}_i, v_{-i}),x,y) -\max_{v'_i}u_i(v_i, (v'_i, v_{-i}),x,y)
=& \hat{u}_i(v_i, (\hat{v}^*_i, v_{-i}),x,y)- u_i(v_i,(v^*_i, v_{-i}),x,y)
\\
\leq& \hat{u}_i(v_i, (\hat{v}^*_i, v_{-i}),x,y)- u_i(v_i,(\hat{v}^*_i, v_{-i}),x,y)
\\
\leq& \max_{v'_i} \Big\vert u_i(v_i, (v'_i, v_{-i}),x,y) - \hat{u}_i(v_i, (v'_i, v_{-i}),x,y)\Big\vert.
\end{aligned}
\end{equation*}

Thus, 
\begin{equation*}
\begin{aligned}
&\Big\vert\max_{v'_i}\big(u_i(v_i, (v'_i, v_{-i})) - u_i(v_i, (v_i, v_{-i}))\big) - \max_{\bar{v}_i}\big(\hat{u}_i(v_i, (\bar{v}_i, v_{-i})) - \hat{u}_i(v_i, (v_i, v_{-i}))\big)\Big\vert \\&
\leq 2\max_{v'_i}\Big\vert u_i(v_i, (v'_i, v_{-i}),x,y) - \hat{u}_i(v_i, (v'_i, v_{-i}),x,y)\Big\vert.
\end{aligned}
\end{equation*}
Summing the inequalities by $i$, this completes the proof that $\mathcal{N}_{\infty,1}(\mathrm{RGT}\circ \mathcal{U}, \epsilon) \leq \mathcal{N}_{\infty,1}(\mathcal{U}, \frac{\epsilon}{2})$.

Next we prove that $\mathcal{N}_{\infty,1}(\mathcal{U}, \epsilon)\leq \mathcal{N}_{\infty,1}(\mathcal{M}, \frac{\epsilon}{n})$.

By the definition of the covering number for the auction class $\Mm$, there exists a cover $\hat{\Mm}$ for $\mathcal{M}$ of size $|\hat{\Mm}| \leq  \mathcal{N}_{\infty, 1}(\mathcal{M},\frac{\epsilon}{n})$
such that for
 any $(g, p)\in \mathcal{M}$, there is a $(\hat{g}, \hat{p}) \in \hat{\Mm}$  for all $v,x,y$, 
\begin{equation*}
\sum_{i,j} \vert g_{ij}(v,x,y) - \hat{g}_{ij}(v,x,y)\vert + \sum_i\vert p_i(v,x,y) - \hat{p}_i(v,x,y)\vert \leq \frac{\epsilon}{n}.
\end{equation*}
For all $v \in \Vv, v'_i \in \Vv_i$,
\begin{equation*}
\begin{aligned}
&\Big\vert u_i(v_i, (v'_i, v_{-i}),x,y) - \hat{u}_i(v_i, (v'_i, v_{-i}),x,y)\Big\vert
\\
\leq& \Big\vert \sum_{j} \big(g_{ij}((v'_i, v_{-i}),x,y) - \hat{g}_{ij}((v'_i, v_{-i}),x,y)\big) v'_{ij}\Big\vert + \Big\vert p_i((v'_i, v_{-i}),x,y) -\hat{p}_i((v'_i, v_{-i}),x,y) \Big\vert
\\
\leq& \sum_j \Big\vert g_{ij}((v'_i, v_{-i}),x,y) - \hat{g}_{ij}((v'_i, v_{-i}),x,y)\Big\vert + \Big\vert p_i((v'_i, v_{-i}),x,y) -\hat{p}_i((v'_i, v_{-i}),x,y) \Big\vert
\\
\leq&\frac{\epsilon}{n}.
\end{aligned}
\end{equation*}
Thus, 
\begin{equation*}
    \sum_{i=1}^{n}\left\vert u_i(v_i, (v'_i, v_{-i}),x,y) - \hat{u}_i(v_i, (v'_i, v_{-i}),x,y)\right\vert\leq n\cdot\frac{\epsilon}{n}=\epsilon.
\end{equation*}
This completes the proof that $\mathcal{N}_{\infty,1}(\mathcal{U}, \epsilon)\leq \mathcal{N}_{\infty,1}(\mathcal{M}, \frac{\epsilon}{n})$

Therefore,
\begin{equation*}
\begin{aligned}
\mathcal{N}_{\infty,1}(\mathrm{RGT}\circ \mathcal{U}, \epsilon) \leq \mathcal{N}_{\infty,1}(\mathcal{U}, \frac{\epsilon}{2})\leq\mathcal{N}_{\infty,1}(\mathcal{M}, \frac{\epsilon}{2n}).
\end{aligned}
\end{equation*}
This completes the proof of \cref{cover2}.
\end{proof}

\subsection{Proof of \cref{thm:UC}}
\begin{proof}[Proof of \cref{thm:UC}]
For all $\epsilon, \delta \in (0, 1)$, when 
\begin{equation*}
    L \ge \frac{9n^2}{2\epsilon^2}\left(\ln\frac{4}{\delta} + \ln{ \Nn_{\infty, 1}(\Mm,\frac{\epsilon}{6n})}\right),
\end{equation*}
Combining \cref{lemma:generalization} and \cref{cover1} together, we get

\begin{equation}
\label{eq:UC:profit_2}
\begin{aligned}
\PP_{S\sim\Dd^m}\Big[&\exists (g^w, p^w) \in \Mm, \bigg|\sum_{i=1}^n\EE_{(v,x,y) \sim\Dd_{v,{x},{y}}}[p^w_i(v, {x}, {y})] - \frac{1}{L}\sum_{i=1}^n \sum_{\ell=1}^L p^w_i(v^{(\ell)}, {x}^{(\ell)}, {y}^{(\ell)})\bigg| > \epsilon\Big]
\\
&\leq 2\Nn_\infty(\Pp,\frac{\epsilon}{3})\exp{(-\frac{2L\epsilon^2}{9n^2})}
\\
&\leq2\Nn_\infty(\Mm,\frac{\epsilon}{3})\exp{(-\frac{2L\epsilon^2}{9n^2})}
\\
&\leq2\Nn_\infty(\Mm,\frac{\epsilon}{6n})\exp{(-\frac{2L\epsilon^2}{9n^2})}
\\
&\le \frac{\delta}{2}.
\end{aligned}
\end{equation}

Similarly, combining \cref{lemma:generalization} and \cref{cover2} together, we have
\begin{equation}
\label{eq:UC:regret_2}
\begin{aligned}
\PP_{S\sim\Dd^m}\Big[&\exists (g^w, p^w) \in \Mm,  {\bigg|\EE_{(v,{x},{y}) \sim \Dd_{v,{x},{y}}}\Big[\sum_{i=1}^n rgt_i(w)\Big]  - 
\sum_{i=1}^n \widehat{rgt}_i(w)\bigg| > \epsilon\Big]}
\\
&\leq 2\Nn_\infty(\mathrm{RGT}\circ\mathcal{U},\frac{\epsilon}{3})\exp{(-\frac{2L\epsilon^2}{9n^2})}
\\
&\leq2\Nn_\infty(\Mm,\frac{\epsilon}{6n})\exp{(-\frac{2L\epsilon^2}{9n^2})}
\\
&\le \frac{\delta}{2}.
\end{aligned}
\end{equation}

Combining \cref{eq:UC:profit_2}, \cref{eq:UC:regret_2} and the Union Bound, with probability at most $\frac{\delta}{2}+\frac{\delta}{2}=\delta$, one of the two events of \cref{eq:UC:profit_2} and \cref{eq:UC:regret_2} happens.
Therefore, with probability at least $1 - \delta$, \cref{eq:UC:profit} and \cref{eq:UC:regret} both hold. We complete the proof of \cref{thm:UC}.

\end{proof}

\end{document}